\documentclass[preprint,a4paper,11pt]{elsarticle}
\usepackage[T1]{fontenc}
\usepackage[utf8]{inputenc}
\usepackage[english]{babel}
\usepackage{amsmath,amsthm,enumerate}
\usepackage{amssymb}
\usepackage{amscd}
\usepackage{vmargin}
\usepackage{url}
\usepackage{stmaryrd}
\usepackage{tikz,pgf}
\usepackage{subfig}
\usepackage{cancel}
\usepackage[textsize=footnotesize]{todonotes}
\usetikzlibrary{positioning}
\usepackage{subfig}
\usetikzlibrary{shapes.geometric}
\usepackage{soul}

\makeatletter
\def\ps@pprintTitle{%
 \let\@oddhead\@empty
 \let\@evenhead\@empty
 \def\@oddfoot{}%
 \let\@evenfoot\@oddfoot}
\makeatother
  
\begin{document}
\begin{frontmatter}

\title{A proof of the Multiplicative 1-2-3 Conjecture}

\author[nice]{Julien Bensmail}
\author[labri]{Herv\'e Hocquard}
\author[labri]{Dimitri Lajou}
\author[labri]{\'Eric Sopena}

\address[nice]{Universit\'e C\^ote d'Azur, CNRS, Inria, I3S, France}
\address[labri]{Univ. Bordeaux, CNRS,  Bordeaux INP, LaBRI, UMR 5800, F-33400, Talence, France}

\journal{...}

\begin{abstract}
We prove that the product version of the 1-2-3 Conjecture, raised by Skowronek-Kazi\'ow in 2012, is true. Namely, for every connected graph with order at least~$3$, we prove that we can assign labels~$1,2,3$ to the edges in such a way that no two adjacent vertices are incident to the same product of labels.
\end{abstract}

 
\end{frontmatter}

\newtheorem{theorem}{Theorem}[section]
\newtheorem{lemma}[theorem]{Lemma}
\newtheorem{conjecture}[theorem]{Conjecture}
\newtheorem{observation}[theorem]{Observation}
\newtheorem{claim}[theorem]{Claim}
\newtheorem{corollary}[theorem]{Corollary}
\newtheorem{proposition}[theorem]{Proposition}
\newtheorem{question}[theorem]{Question}
\newtheorem{definition}[theorem]{Definition}
\newtheorem*{123c}{1-2-3 Conjecture (sum version)}
\newtheorem*{m123c}{1-2-3 Conjecture (multiset version)}
\newtheorem*{p123c}{1-2-3 Conjecture (product version)}

\newcommand{\qedclaim}{\hfill $\diamond$ \medskip}
\newenvironment{proofclaim}{\noindent{\em Proof of the claim.}}{\qedclaim}

\newcommand{\chis}{\chi_{\rm S}}
\newcommand{\chim}{\chi_{\rm M}}
\newcommand{\chip}{\chi_{\rm P}}
\newcommand{\set}[1]{\left\{#1\right\}}
\newcommand{\abs}[1]{\left|#1\right|}
\newcommand{\paren}[1]{\left(#1\right)}
\newcommand{\CombNull}{Combinatorial Nullstellensatz}


\section{Introduction}

Let $G$ be a graph.
A \textit{$k$-labelling} $\ell: E(G) \rightarrow \{1,\dots,k\}$ is an assignment of labels~$1,\dots,k$ to the edges of $G$. From $\ell$, we can compute different parameters of interest for all vertices $v$, such as the \textit{sum} $\sigma_\ell(v)$ of incident labels (being formally $\sigma_\ell(v)=\Sigma_{u \in N(v)} \ell(uv)$), or similarly the \textit{multiset} $\mu_\ell(v)$ of labels incident to $v$ or the \textit{product} $\rho_\ell(v)$ of labels incident to $v$.
We say that $\ell$ is \textit{s-proper} if $\sigma_\ell$ is a proper vertex-colouring of $G$, \textit{i.e.}, we have $\sigma_\ell(u) \neq \sigma_\ell(v)$ for every edge $uv \in E(G)$. Similarly, we say that $\ell$ is \textit{m-proper} and \textit{p-proper}, if $\mu_\ell$ and $\rho_\ell$, respectively, form proper vertex-colourings of $G$.

In the context of so-called \textit{distinguishing labellings}, the goal is generally to not only distinguish vertices within some distance according to some parameter computed from labellings (such as the parameters $\sigma_\ell$, $\mu_\ell$ and $\rho_\ell$ above, to name a few), but also to construct such $k
$-labellings with $k$ being as small as possible. We refer the interested reader to~\cite{Gal98}, in which hundreds of such labelling techniques are listed. 

Regarding s-proper, m-proper and p-proper labellings, which are the main focus in this work, we are thus interested, as mentioned above, in finding such $k$-labellings with $k$ as small as possible, for a given graph $G$. In other words, we are interested in the parameters $\chis(G)$, $\chim(G)$ and $\chip(G)$ which denote the smallest $k \geq 1$ such that s-proper, m-proper and p-proper, respectively, $k$-labellings exist (if any). Actually, through greedy labelling arguments, it can be observed that the only connected graph $G$ for which $\chis(G)$, $\chim(G)$ or $\chip(G)$ is not defined, is $K_2$, the complete graph on $2$ vertices. Consequently, these three parameters are generally investigated for so-called \textit{nice graphs}, which are those graphs with no connected component isomorphic to $K_2$.

S-proper, m-proper and p-proper labellings form a subfield of distinguishing labellings, which has been attracting attention due to the so-called 1-2-3 Conjecture, raised, in~\cite{KLT04}, by Karo\'nski, {\L}uczak and Thomason in 2004:

\begin{123c}
If $G$ is a nice graph, then $\chis(G) \leq 3$.
\end{123c}

Later on, counterparts of the 1-2-3 Conjecture were raised for m-proper and p-proper labellings. 
Addario-Berry, Aldred, Dalal and Reed first raised, in 2005, the following in~\cite{AADR05}:

\begin{m123c}
If $G$ is a nice graph, then $\chim(G) \leq 3$.
\end{m123c}

\noindent while Skowronek-Kazi\'ow then raised, in 2012, the following in~\cite{SK12}:

\begin{p123c}
If $G$ is a nice graph, then $\chip(G) \leq 3$.
\end{p123c}

It is worth mentioning that all three conjectures above, if true, 
would be tight, as attested for instance by complete graphs.
Note also that the multiset version of the 1-2-3 Conjecture is, out of the three variants above, the easiest one in a sense, as every s-proper or p-proper labelling is also m-proper (thus, proving the sum or product variant of the 1-2-3 Conjecture would prove the multiset variant).

To date, the best result towards the sum version of the 1-2-3 Conjecture, proved by Kalkowski, Karo\'nski and Pfender in~\cite{KKP10},
is that $\chis(G) \leq 5$ holds for every nice graph $G$.
Another significant result is due to Przyby{\l}o, who recently proved in~\cite{Prz19c} that even $\chis(G) \leq 4$ holds for every nice regular graph $G$.
Karo\'nski, {\L}uczak and Thomason themselves also proved in~\cite{KLT04} that $\chis(G) \leq 3$ holds for nice $3$-colourable graphs.
Regarding the multiset version, for long the best result was the one proved by Addario-Berry, Aldred, Dalal and Reed in~\cite{AADR05}, 
stating that $\chim(G) \leq 4$ holds for every nice graph $G$.
Building on that result, Skowronek-Kazi\'ow later proved in~\cite{SK12} that $\chip(G) \leq 4$ holds for every nice graph $G$.
She also proved that $\chip(G) \leq 3$ holds for every nice $3$-colourable graph $G$.

A breakthrough result was recently obtained by Vu\v{c}kovi\'c, 
as he totally proved the multiset version of the 1-2-3 Conjecture in~\cite{Vuc18}.
Due to connections between m-proper and p-proper $3$-labellings,
we observed in~\cite{BHLS20} that this result directly implies that $\chip(G) \leq 3$ holds for every nice regular graph $G$.
Inspired by Vu\v{c}kovi\'c's proof scheme, we were also able to prove that $\chip(G) \leq 3$ holds for nice $4$-chromatic graphs $G$, and to prove related results that are very close to what is stated in the product version of the 1-2-3 Conjecture. 

Building on these results, we prove the following  throughout the rest of this paper.

\begin{theorem}\label{main-theorem}
The product version of the 1-2-3 Conjecture is true.
That is, every nice graph admits p-proper $3$-labellings.
\end{theorem}


\section{Proof of Theorem~\ref{main-theorem}}

Let us start by introducing some terminology and recalling some properties of p-proper labellings, which will be used throughout the proof.
Let $G$ be a graph, and $\ell$ be a $3$-labelling of $G$.
For a vertex $v \in V(G)$ and a label $i \in \{1,2,3\}$, we denote by $d_i(v)$ the \textit{$i$-degree} of $v$ by $\ell$, being the number of edges incident to $v$ that are assigned label~$i$ by $\ell$. 
Note then that $\rho_\ell(v)=2^{d_2(v)} 3^{d_3(v)}$. 
We say that $v$ is \textit{$1$-monochromatic} if $d_2(v)=d_3(v)=0$, while we say that $v$ is \textit{$2$-monochromatic} (\textit{$3$-monochromatic}, resp.) if $d_2(v)>0$ and $d_3(v)=0$ ($d_3(v)>0$ and $d_2(v)=0$, resp.). In case $v$ has both $2$-degree and $3$-degree at least~$1$, we say that $v$ is \textit{bichromatic}. We also define the \textit{$\{2,3\}$-degree} of $v$ as the sum $d_2(v)+d_3(v)$ of its $2$-degree and its $3$-degree. Thus, if $v$ is bichromatic, then its $\{2,3\}$-degree is at least~$2$.

Because $\ell$ assigns labels $1,2,3$, and, in particular,
because $2$ and $3$ are coprime,
note that, for every edge $uv$ of $G$, we have $\rho_\ell(u) \neq \rho_\ell(v)$ as soon as $u$ and $v$ have different $2$-degrees, $3$-degrees, or $\{2,3\}$-degrees. In particular, $u$ and $v$ cannot be in conflict, \textit{i.e.}, verify $\rho_\ell(u)=\rho_\ell(v)$, if $u$ and $v$ are $i$-monochromatic and $j$-monochromatic, respectively, for $i \neq j$, or if $u$ is monochromatic while $v$ is bichromatic.

\medskip

Before going into the proof of Theorem~\ref{main-theorem}, let us start by giving an overview of it.
Let $G$ be a nice graph. Our goal is to build a p-proper $3$-labelling $\ell$ of $G$. We can clearly assume that $G$ is connected. We also set $t=\chi(G)$, where, recall, $\chi(G)$ refers to the chromatic number\footnote{Recall that a \textit{proper $k$-vertex-colouring} of a graph $G$ is a partition $(V_1, \dots, V_k$) of $V(G)$ where all $V_i$'s are independent. The \textit{chromatic number} $\chi(G)$ of $G$ is the smallest $k \geq 1$ such that proper $k$-vertex-colourings of $G$ exist. We say that $G$ is \textit{$k$-colourable} if $\chi(G) \leq k$.} of $G$. In particular, $t \geq 2$. We could even assume that $t \geq 5$, due to the product version of the 1-2-3 Conjecture being true for $4$-colourable graphs (recall~\cite{BHLS20}), though this is not needed throughout the proof.

In what follows, we construct $\ell$ through three main steps.
First, we need to partition the vertices of $G$ in a way verifying specific cut properties, forming what we call a valid partition of $V(G)$ (see later Definition~\ref{definition:valid-partition} for a more formal definition).
A \textit{valid partition} $\mathcal{V}=(V_1,\dots,V_t)$ is a partition of $V(G)$ into $t$ independent sets $V_1,\dots,V_t$ fulfilling two main properties, being, roughly put, that 1) every vertex $v$ in some part $V_i$ with $i>1$ has an incident \textit{upward edge} to every part $V_j$ with $j<i$, and 2) for every connected component of $G[V_1 \cup V_2]$ having only one edge, we can freely swap its two vertices in $V_1$ and $V_2$ while preserving the properties of a valid partition.

Once we have this valid partition $\mathcal{V}$ in hand, we can then start constructing $\ell$. The main part of the labelling process, Step~2 below, consists in starting from all edges of $G$ being assigned label~$1$ by $\ell$, and then processing the vertices of $V_3,\dots,V_t$ one after another, possibly changing the labels assigned by $\ell$ to some of their incident edges, so that certain product types are achieved by~$\rho_\ell$. These desired product types can be achieved due to the many upward edges that some vertices are incident to (in particular, the deeper a vertex lies in $\mathcal{V}$, the more upward edges it is incident to).
The product types we achieve for the vertices depend on the part $V_i$ of $\mathcal{V}$ they belong to. In particular, the modifications we make on $\ell$ guarantee that all vertices in $V_3, \dots, V_t$ are bichromatic, every two vertices in $V_i$ and $V_j$ with $i,j \in \{3,\dots,t\}$ and $i \neq j$ have distinct $2$-degrees or $3$-degrees, all vertices in $V_2$ are $1$-monochromatic or $2$-monochromatic, and all vertices in $V_1$ are $1$-monochromatic or $3$-monochromatic. By itself, achieving these product types makes $\ell$ almost p-proper, in the sense that the only possible conflicts are between $1$-monochromatic vertices in $V_1$ and $V_2$. An important point also, is that, through these label modifications, we will make sure that all edges of $G[V_1 \cup V_2]$ remain assigned label~$1$, and no vertex in $V_3 \cup \dots \cup V_t$ has $3$-degree~$1$, $2$-degree at least~$2$, and odd $\{2,3\}$-degree; in last Step~3 below, we will use that last fact to remove remaining conflicts by allowing some vertices of $V_1 \cup V_2$ to become \textit{special}, \textit{i.e.}, make their product realising these exact label conditions.

Step~3 is designed to get rid of the last conflicts between the adjacent $1$-monochromatic vertices of $V_1$ and $V_2$ without introducing new ones in $G$. To that end, we will consider the set $\mathcal{H}$ of the connected components of $G[V_1 \cup V_2]$ having conflicting vertices, and, if needed, modify the labels assigned by $\ell$ to some of their incident edges so that no conflicts remain, and no new conflicts are created in $G$. To make sure that no new conflicts are created between vertices in $V_1 \cup V_2$ and vertices in $V_3 \cup \dots \cup V_t$, we will modify labels while making sure that all vertices in $V_1 \cup V_2$ are monochromatic or special. An important point also, is that the fixing procedures we introduce require the number of edges in a connected component of $\mathcal{H}$ to be at least~$2$. Because of that, once Step~2 ends, we must make sure that $\mathcal{H}$ does not contain a connected component with only one edge incident to two $1$-monochromatic vertices. To guarantee this, we will also make sure, during Step~2, to modify labels and the partition $\mathcal{V}$ slightly so that $\mathcal{H}$ has no such configuration.

\subsection*{Step 1: Constructing a valid partition}\label{sec:123:step0}

Let $\mathcal V = (V_1, \dots, V_t)$ be a partition of $V(G)$ where each $V_i$ is an independent set. Note that such a partition exists, as, for instance, any proper $t$-vertex-colouring of $G$ forms such a partition of $V(G)$. For every vertex $u \in V_i$, an incident \textit{upward edge} (\textit{downward edge}, resp.) is an edge $uv$ for which $v$ belongs to some $V_j$ with $j<i$ ($j>i$, resp.). Note that all vertices in~$V_1$ have no incident upward edges, while all vertices in $V_t$ have no incident downward edges.

We denote by $M_0(\mathcal V)$ (also denoted $M_0$ when the context is clear) the set of isolated edges in the subgraph $G[V_1 \cup V_2]$ of $G$ induced by the vertices of $V_1 \cup V_2$.
That is, $M_0$ contains the edges of the connected components of $G[V_1 \cup V_2]$ that consist of one edge only. 
To lighten the exposition, whenever referring to the vertices of $M_0$, we mean the vertices of $G$ incident to the edges in $M_0$.

For an edge $uv \in M_0$ with $u \in V_1$ and $v \in V_2$, \textit{swapping} $uv$ consists in modifying the partition $\mathcal{V}$ by removing $u$ from $V_1$ ($v$ from $V_2$, resp.) and adding it to $V_2$ ($V_1$, resp.). In other words, we exchange the parts to which $u$ and $v$ belong. Note that if $V_1$ and $V_2$ are independent sets before the swap, then, because $uv \in M_0$, by definition the resulting new $V_1$ and $V_2$ remain independent. 
Also, the set $M_0$ is unchanged by the swap operation.

We can now give a formal definition for the notion of valid partition.

\begin{definition}[Valid partition]\label{definition:valid-partition}
For a $t$-colourable graph $G$,
a partition $\mathcal V = (V_1, \dots, V_t)$ of $V(G)$ is a \emph{valid partition} (of $G$) if $\mathcal V$ verifies the following properties.
\begin{enumerate}
    \item[$(\mathcal{I})$] Every $V_i$ is an independent set.
    \item[$(\mathcal P_1)$] Every vertex in every $V_i$ with $i \geq 2$ has a neighbour in $V_j$ for every $j <i$.
    \item[$(\mathcal{S})$] For every set $\{e_1,\dots,e_p\}$ of edges of $M_0(\mathcal V)$, successively swapping every $e_i$ (in any order) results in a partition $\mathcal V'$ verifying Properties $(\mathcal{I})$ and $(\mathcal P_1)$.
\end{enumerate}
\end{definition}

Note that Property $(\mathcal{S})$ in Definition~\ref{definition:valid-partition} implies that any valid partition $\mathcal V$ also verifies the following additional property:

\begin{enumerate}
    \item[$(\mathcal P_2)$] Successively swapping any number of edges of $M_0(\mathcal V)$ results in a valid partition $\mathcal V'$.
\end{enumerate}

In order to prove Theorem~\ref{main-theorem}, as mentioned earlier, to start constructing $\ell$ we first need to have a valid partition of $G$ in hand.
The following result guarantees its existence.

\begin{lemma}\label{lem:validpartition}
Every nice $t$-colourable graph $G$ admits a valid partition.
\end{lemma}

\begin{proof}
For a partition $\mathcal V = (V_1,\dots, V_t)$ of $V(G)$ where each $V_i$ is an independent set (note that such a partition exists, as every proper $t$-vertex-colouring of $G$ is one such partition), set $f(\mathcal V) = \sum_{k=1}^t k \cdot \abs{V_k}$. Among all possible $\mathcal{V}$'s, we consider a $\mathcal V$ that minimises $f(\mathcal V)$.  

Suppose that there is a vertex $u \in V_i$ with $i \geq 2$ for which Property $(\mathcal P_1)$ does not hold, \textit{i.e.}, there is a $j<i$ such that $u$ has no incident upward edge to $V_j$. By moving $u$ to $V_j$, we obtain another partition $\mathcal V'$ of $V(G)$ where every part is an independent set. However, note that $f(\mathcal V') = f(\mathcal V) + j - i < f(\mathcal V)$, a contradiction to the minimality of $\mathcal V$. From this, we deduce that every partition $\mathcal{V}$ minimising $f$ must verify Property~$(\mathcal P_1)$.

Let now $\mathcal V'$ be the partition of $V(G)$ obtained by successively swapping edges of $M_0(\mathcal V)$. Recall that the swapping operation preserves Property~$(\mathcal I)$ and observe that $f(\mathcal V) = f(\mathcal V')$. Hence, $\mathcal V'$ minimises $f$ and thus verifies Properties~$(\mathcal I)$ and~$(\mathcal P_1)$. Thus Property~$(\mathcal S)$ also holds, and $\mathcal V$ is a valid partition of $G$.
\end{proof}

From here, we thus assume that we have a valid partition $\mathcal{V}=(V_1, \dots, V_t)$ of $G$.


\subsection*{Step 2: Labelling the upward edges of $V_3, \dots, V_t$}\label{sec:123:step1}

From $G$ and $\mathcal V$, our goal now is to construct a $3$-labelling $\ell$ of $G$ achieving certain properties, the most important of which being that the only possible conflicts are between pairs of vertices of $V_1$ and $V_2$ that do not form an edge of $M_0$. The following result sums up the exact conditions we want $\ell$ to fulfil. Recall that a vertex $v$ is special by $\ell$, if $d_3(v)=1$, $d_2(v) \geq 2$ and $d_2(v)+d_3(v)$ is odd.
Note that special vertices are bichromatic.

\begin{lemma}\label{lem:123}
For every nice graph $G$ and every valid partition $(V_1,\dots,V_t)$ of $G$, there exists a $3$-labelling $\ell$ of $G$ such that:
\begin{enumerate}
    \item all vertices of $V_1$ are either $1$-monochromatic or $3$-monochromatic,
    \item all vertices of $V_2$ are either $1$-monochromatic or $2$-monochromatic,
    \item all vertices of $V_3 \cup \dots \cup V_t$ are bichromatic,
    \item no vertex is special,
    \item if $u \in V_1$ and $v \in V_2$ are adjacent, then $\ell(uv) = 1$,
    \item if two vertices $u$ and $v$ are in conflict, then $u \in V_1$ and $v \in V_2$ (or \textit{vice versa}), and at least one of them has a neighbour $w$ in $V_1 \cup V_2 \setminus \{u,v\}$.
\end{enumerate}
\end{lemma}

\begin{proof}
From now on, we fix the valid partition $\mathcal{V}=(V_1,\dots,V_t)$ of $G$. During the construction of $\ell$, we may have, however, to swap some edges of $M_0$, resulting in a different valid partition of $G$. Abusing the notations, for simplicity we will still denote by $\mathcal V$ any valid partition of $G$ obtained this way, through swapping edges. Recall that valid partitions are closed under swapping edges of $M_0$ (by Property~$(\mathcal P_2)$).

Our goal is to design $\ell$ so that it not only verifies the four colour properties of Items~1 to 4 of the statement, but also achieves the following refined product types, for every vertex $v$ in a part $V_i$ of $\mathcal V$:

\begin{itemize}
    \item $v \in V_1$: $v$ is $1$-monochromatic or $3$-monochromatic;
    \item $v \in V_2$: $v$ is $1$-monochromatic or $2$-monochromatic;
    
    \item $v \in V_3$: $v$ is bichromatic with $2$-degree~$1$ and even $\{2,3\}$-degree;
    
    \item $v \in V_4$: $v$ is bichromatic with $3$-degree~$2$ and odd $\{2,3\}$-degree;
    \item $v \in V_5$: $v$ is bichromatic with $2$-degree~$2$ and even $\{2,3\}$-degree;
    
    \item ...
    
    \item $v \in V_{2n}$, $n \geq 3$: $v$ is bichromatic with $3$-degree~$n$ and odd $\{2,3\}$-degree;
    \item $v \in V_{2n+1}$, $n \geq 3$: $v$ is bichromatic with $2$-degree~$n$ and even $\{2,3\}$-degree;
    
    \item ...
\end{itemize}

We start from $\ell$ assigning label~$1$ to all edges of $G$.
Let us now describe how to modify $\ell$ so that the conditions above are met for all vertices. We consider the vertices of $V_t,\dots,V_3$ following that order, from ``bottom to top'', and modify labels assigned to upward edges. An important condition we will maintain, is that every vertex in an odd part $V_{2n+1}$ ($n \geq 0$) has all its incident downward edges (if any) labelled~$3$ or~$1$, while every vertex in an even part $V_{2n}$ ($n \geq 1$) has all its incident downward edges (if any) labelled~$2$ or~$1$. Note that this is trivially verified for the vertices in $V_t$, since they have no incident downward edges.

At any point in the process, let $M$ be the set of edges of $M_0$ for which both ends are $1$-monochromatic (initially, $M=M_0$). When treating a vertex $u \in V_3 \cup \dots \cup V_t$, we define $M_u$ as the subset of edges of $M$ having an end that is a neighbour of $u$. For every edge $e\in M_u$, we choose one end of $e$ that is a neighbour of $u$ and we add it to a set $S_u$. Note that $\abs{S_u} = \abs{M_u}$. Another goal during the labelling process, to fulfil Item~6, is to label the edges incident to $u$ so that at least one end of every edge in $M_u$ is no longer $1$-monochromatic. Note that the set $M_u$ considered when labelling the edges incident to $u$ is not necessarily the set of edges of $M_0$ incident to a neighbour of $u$, as, during the whole process, some of these edges might be removed from $M$ when dealing with previous vertices in $V_3 \cup \dots \cup V_t$.

\medskip

Let us now consider the vertices in $V_t, \dots, V_3$ one by one, following that order.
Let thus $u \in V_{i}$ be a vertex that has not been treated yet, with $i \geq 3$. Recall that every vertex belonging to some $V_j$ with $j > i$ was treated earlier on, and thus has its desired product. Suppose that $i = 2n$ with $n \geq 2$ ($i=2n+1$ with $n \geq 1$, resp.). Recall also that $u$ is assumed to have all its incident downward edges labelled~$1$ or~$2$ ($3$, resp.), due to how vertices in $V_j$'s with $j>i$ have been treated earlier on.
Also, all upward edges incident to $u$ are currently assigned labelled~$1$ by $\ell$.

If $M_u \neq \varnothing$, then we swap edges of $M_u$, if necessary, so that every vertex in $S_u$ belongs to $V_2$ ($V_1$, resp.). This does not invalidate any of our invariants since both ends of an edge in $S_u$ are $1$-monochromatic.

In any case, by Property $(\mathcal P_1)$, we know that, for every $j < i$, there is a vertex $x_j \in V_j$ which is a neighbour of $u$. In particular, the vertex $x_1$ ($x_2$, resp.) does not belong to $S_u$ (but may be the other end of an edge in $M_u$). We label the edges $ux_3,ux_5,\dots,ux_{2n-1}$ with~$3$ ($ux_4,ux_6,\dots,ux_{2n}$ with~$2$, resp.). Note that, at this point, $d_3(u) = n-1$ ($d_2(u) = n-1$, resp.). 
To finish dealing with $u$, 
we need to distinguish two cases depending on whether $M_u$ is empty or not.

\begin{itemize}
    \item Suppose first that $M_u = \varnothing$. Label $ux_1$ with~$3$ ($ux_2$ with~$2$, resp.). Now $u$ has the desired $3$-degree ($2$-degree, resp.). If $i > 3$, then label $ux_{i-2}$ with~$2$ ($3$, resp.) so that $u$ is sure to be bichromatic. If $i > 3$ and the $\set{2,3}$-degree of $u$ does not have the desired parity, then label $ux_2$ with~$2$ ($ux_1$ with~$3$, resp.). If $u \in V_3$ and the $\set{2,3}$-degree of $u$ is even, then $u$ is already bichromatic since $d_2(u) = 1$. If $u \in V_3$ and the $\set{2,3}$-degree of $u$ is odd, then label $ux_1$ with~$3$ to adjust the parity of the $\set{2,3}$-degree of $u$ and make $u$ bichromatic. In all cases, at this point $u$ is bichromatic with $3$-degree~$n$ ($2$-degree~$n$, resp.) and odd $\set{2,3}$-degree (even $\set{2,3}$-degree, resp.), which is precisely what is desired for $u$.
    
    \item Suppose now that $M_u \neq \varnothing$. Let $z \in S_u$ and let $e$ be the edge of $M_u$ containing $z$. For every vertex $w \in S_u \setminus \set{z}$, we label the edge $uw$ with~$2$ ($3$, resp.). 
    Then:
    \begin{itemize}
        \item If $d_2(u) + d_3(u)$ is odd (even, resp.), then label $uz$ with~$2$ ($3$, resp.) and $ux_1$ with~$3$ ($ux_2$ with~$2$, resp.). In this case, every edge in $M_u$ is incident to at least one vertex which is not $1$-monochromatic, while $u$ is bichromatic with $3$-degree $n$ ($2$-degree $n$, resp.)  and odd $\set{2,3}$-degree (even $\set{2,3}$-degree, resp.).
    
        \item If $d_2(u) + d_3(u)$ is even (odd, resp.) and $d_2(u) > 0$ ($d_3(u) > 0$, resp.), then swap $e$ and label $uz$ with~$3$ ($2$, resp.). Note that, after the swap of $e$, we have $z \in V_1$ ($z \in V_2$, resp.). In this case, every edge in $M_u$ is incident to at least one vertex which is not $1$-monochromatic, while $u$ is bichromatic with $3$-degree $n$ ($2$-degree $n$, resp.)  and odd $\set{2,3}$-degree (even $\set{2,3}$-degree, resp.).

        \item The last case is when $d_2(u) + d_3(u)$ is even (odd, resp.) and $d_2(u) = 0$ ($d_3(u) = 0$, resp.). If $i > 4$, then we can label $ux_{i-2}$ with~$2$ ($3$, resp.) and fall back into one of the previous cases.    If $i = 4$, then the only edge labelled~$3$ is the edge $ux_3$ which implies that $d_3(u) = 1$, which is impossible since $d_2(u) = 0$ and thus $d_2(u) + d_3(u)$ is odd which contradicts our hypothesis. If $i=3$, then the conditions of this case imply that $d_2(u) \geq 1$ while every upward edge incident to $u$ is labelled~$1$ or~$3$ and similarly for every incident downward edge; this case thus cannot occur.
    \end{itemize}
    
    To finish, we remove the edges of $M_u$ from $M$ since their two ends are not both $1$-monochromatic anymore.
\end{itemize}

At the end of this process, all vertices in $V_1$ are $1$-monochromatic or $3$-monochromatic, while all vertices in $V_2$ are $1$-monochromatic or $2$-monochromatic.
Every vertex in $V_3 \cup \dots \cup V_t$ is bichromatic and there are no conflicts involving any pair of these vertices. Indeed if $a\in V_i$ and $b \in V_j$ are adjacent with $i> j \geq 3$, then either $i$ and $j$ do not have the same parity, in which case $a$ and $b$ do not have the same $\set{2,3}$-degree; or both $i$ and $j$ are even (odd, resp.) and $d_3(a) = \frac{i}{2} \neq \frac{j}{2} = d_3(b)$ ($d_2(a) = \frac{i-1}{2} \neq \frac{j-1}{2} = d_2(b)$, resp.).  Note also that no vertex in $G$ is special, as, by definition, special vertices have $3$-degree~$1$, $2$-degree at least~$2$, and odd $\set{2,3}$-degree. Moreover, we did not change the label of any edge in the cut $(V_1,V_2)$.

Finally, suppose that there is a conflict between two vertices $u$ and $v$. Previous remarks imply that $u \in V_1$ and $v \in V_2$ (or \textit{vice versa}) and that both $u$ and $v$ are $1$-monochromatic. If none of $u$ and $v$ has another neighbour $w$ in $V_1 \cup V_2$, then the edge $uv$ belongs to the set $M_0$. Since $G$ is nice, one of $u$ or $v$ must have a neighbour in $V_3 \cup \dots \cup V_t$. Hence $uv \in M_z$ for one such neighbour $z$. Recall also that we relabelled the edges incident to $z$ in such a way that, for every edge of $M_z$, at least one incident vertex became $2$-monochromatic or $3$-monochromatic, a contradiction to the existence of $u$ and $v$. Hence, all properties of the lemma hold.
\end{proof}

\subsection*{Step 3: Labelling the edges between $V_1$ and $V_2$}\label{sec:123:step2}

From now on, we will modify a $3$-labelling $\ell$ of $G$ obtained by applying Lemma~\ref{lem:123}.
We denote by $\mathcal H$ the set of the connected components of $G[V_1 \cup V_2]$ that contain two adjacent vertices $u \in V_1$ and $v \in V_2$ having the same product by $\ell$. By Items~1 and~2 of Lemma~\ref{lem:123}, such $u$ and $v$ are $1$-monochromatic.
Also, by Item~6 of Lemma~\ref{lem:123}, recall that every connected component of $\mathcal{H}$ has at least two edges.
In what follows, we only relabel edges of some connected components $H \in \mathcal H$ while making sure that their vertices (in $V_1 \cup V_2$) are monochromatic or special. This ensures that only vertices of $H$ have their product affected, thus no new conflicts involving vertices in $V_3 \cup \dots \cup V_t$ are created. 

For a subgraph $X$ of $H \in \mathcal{H}$ (possibly $X = H$), if, after having relabelled edges of $X$, no conflict remains between vertices of $X$ and all vertices of $X$ are either monochromatic or special, then we say that $X$ verifies Property~$(\mathcal P_3)$.

\begin{lemma}\label{lem:123:P3}
If we can relabel the edges of every $H \in \mathcal H$ so that every $H$ verifies Property~$(\mathcal P_3)$, then the resulting $3$-labelling is p-proper.
\end{lemma}

\begin{proof}
This is because if we get rid of all conflicts in $\mathcal{H}$, then the only possible remaining conflicts are between vertices in $V_1 \cup V_2$ and in $V_3 \cup \dots \cup V_t$. In particular, recall that any two vertices of two distinct connected components $H_1,H_2 \in G[V_1 \cup V_2]$ cannot be adjacent.
Note also that, because we only relabelled edges in $\mathcal{H}$, the vertices in $V_3 \cup \dots \cup V_t$ retain the product types described in Lemma~\ref{lem:123}. In particular, they remain bichromatic and none of them is special. Thus, they cannot be in conflict with the vertices in $V_1 \cup V_2$.
\end{proof}

In order to show that we can relabel the edges of every $H \in \mathcal H$ so that it fulfils Property $(\mathcal P_3)$, the following result will be particularly handy.

\begin{lemma}\label{lemma:bipartite}
For every integer $s \in \set{2,3}$, every connected bipartite graph $H$ whose edges are labelled $1$ or $s$, and 
 any vertex $v$ in any part $V_i \in \{V_1,V_2\}$ of $H$, we can relabel the edges of $H$ with $1$ and $s$ so that $d_s(u)$ is odd (even, resp.) for every $u \in V_i \setminus \{v\}$, and $d_s(u)$ is even (odd, resp.) for every $u \in V_{3-i}$. 
\end{lemma}

\begin{proof}
As long as $H$ has a vertex $u$ different from $v$ that does not verify the desired condition, apply the following.
Choose any path $P$ from $u$ to $v$, which exists by the connectedness of $H$.
Now follow $P$ from $u$ to $v$, and change the labels of the traversed edges from~$1$ to~$s$ and \textit{vice versa}. 
It can be noted that this alters the parity of the $s$-degrees of $u$ and $v$,
while this does not alter that parity for any of the other vertices of $H$.
Thus, this makes $u$ satisfy the desired condition, while the situation did not change for the other vertices different from $u$ and $v$.
Thus, once this process ends, all vertices of $H$ different from $v$ have their $s$-degree being as desired by the resulting labelling.
\end{proof}

We are now ready to treat the connected components $H \in \mathcal{H}$ independently, so that they all meet Property~$(\mathcal P_3)$.
To ease the reading, we distinguish several cases depending on the types and on the degrees of the vertices that $H$ includes. In each of the successive cases we consider, it is implicitly assumed that $H$ does not meet the conditions of any previous case.

\begin{claim}\label{claim:123:noblue}
If $H \in \mathcal{H}$ contains a $3$-monochromatic vertex $v \in V_1$,
or a $1$-monochromatic vertex $v_1 \in V_1$ having two  $1$-monochromatic neighbours $u_1,u_2 \in V_2$ with degree~$1$ (in $H$), then we can relabel edges of $H$ so that $H$ verifies Property~$(\mathcal P_3)$.
\end{claim}

\begin{proof}
Recall that all edges of $H$ (and thus in $\mathcal H$) are assigned label~$1$; thus, if a vertex of $H$ is $3$-monochromatic, then it must be due to incident downward edges to $V_3, \dots, V_t$.

If $H$ has a $1$-monochromatic vertex $v_1 \in V_1$  having two degree-$1$ $1$-monochromatic neighbours $u_1,u_2 \in V_2$, then we set $\ell(v_1u_1) = \ell(v_1u_2) = 3$. Note that $u_1$ and $u_2$ become $3$-monochromatic with $3$-degree~$1$, and are thus no longer in conflict with $v_1$, as it becomes $3$-monochromatic with $3$-degree~$2$.
Note that either we got rid of all conflicts in $H$ and $H$ now verifies Property~$(\mathcal P_3)$ as desired, or conflicts between other $1$-monochromatic vertices of $H$ remain.
In the latter case, we continue with the following arguments.

Assume $H$ has remaining conflicts, and that $H$ has a $3$-monochromatic vertex $v \in V_1$ (and, due to the previous process, perhaps $3$-monochromatic vertices $u_1$ and $u_2$ in $V_2$, in which case their $3$-degree (and degree in $H$) is precisely~$1$, while their unique neighbour $v$ in $V_1 \cap V(H)$ is $3$-monochromatic with $3$-degree~$2$).
Let $X$ be the set of all $3$-monochromatic vertices of $H$ belonging to $V_1$. Let $C_1,\dots,C_q$ denote the $q \geq 1$ connected components of $H-X$ that do not contain any $3$-monochromatic vertex of $V_2$ (the vertices $u_1$ and $u_2$ we dealt with earlier on). For every $C_i$, we choose arbitrarily a vertex $x_i \in X$ and a vertex $y_i \in C_i$ such that $x_i$ and $y_i$ are adjacent in $H$. Note that the vertices of $C_i$ are either $1$-monochromatic or $2$-monochromatic (in which case they belong to $V_2$), since all $3$-monochromatic vertices of $H$ are part of $X$ (or are the vertices $u_1$ and $u_2$ dealt with earlier on, which we have omitted for now and are not part of the $C_i$'s).

By Lemma~\ref{lemma:bipartite}, in every $C_i$ we can relabel the edges with~$1$ and~$2$ so that all vertices in $(V_2 \cap V(C_i)) \setminus \{y_i\}$ are $2$-monochromatic with odd $2$-degree, while all vertices in $V_1 \cap V(C_i)$ are $2$-monochromatic with even $2$-degree or possibly $1$-monochromatic if their even $2$-degree is~$0$. In particular, recall that $y_i$ must be $1$-monochromatic or $2$-monochromatic. If $y_i$ has odd $2$-degree, then there are no conflicts between vertices of $C_i$. If $y_i$ has even non-zero $2$-degree, then we set $\ell(x_iy_i)=3$, thereby making $y_i$ special.

Let $Y$ be the set containing all $1$-monochromatic $y_i$'s having a $1$-monochromatic neighbour $w_i$ in $C_i$. Let $H'$ be the subgraph of $H$ induced by $Y \cup X$. Note that every edge of $H'$ is labelled~$1$.
Let now $Q_1, \dots, Q_p$ denote the connected components of $H'$ and choose $x_k \in X \cap V(Q_k)$ for every $k \in \set{1,\dots,p}$. For every $k$, we apply  Lemma~\ref{lemma:bipartite} with labels~$1$ and~$3$ so that all vertices in $V_2 \cap V(Q_k)$ get $3$-monochromatic with odd $3$-degree, while all vertices in $V_1 \cap V(Q_k) \setminus \{x_k\}$ get $3$-monochromatic with even $3$-degree or possibly $1$-monochromatic if their $3$-degree is~$0$.

If $x_k$ is involved in a conflict with a vertex $y_i \in V_2 \cap V(Q_k)$, then this is because $x_k$ has odd $3$-degree. Then:

\begin{itemize}
    \item If $\ell(x_ky_i) = 3$, then $d_3(y_i) = d_3(x_k) \geq 3$ since $x_k\in X$ ($x_k$ must thus be incident to at least one other edge labelled $3$, either a downward edge to $V_3,\dots,V_t$ or an edge incident to $u_1$ (and similarly an edge incident to $u_2$)). We here assign label~$1$ to the edge $x_ky_i$ and label~$3$ to the edge $y_iw_i$. This way, $x_k$ gets even $3$-degree while the $3$-degree of $y_i$ does not change. Note that $y_i$ and $w_i$ are not in conflict since $d_3(w_i) = 1$ and $d_3(y_i) \geq 3$. 

    \item Otherwise, if $\ell(x_ky_i) = 1$, then we assign label~$3$ to the edge $x_ky_i$ and label~$3$ to the edge $y_iw_i$. This way, $x_k$ gets even $3$-degree while the $3$-degree of $y_i$ remains odd and must be at least $3$. Again $y_i$ and $w_i$ are not in conflict since $d_3(w_i) = 1$ and $d_3(y_i) \geq 3$. 
\end{itemize}

We claim that we got rid of all conflicts in $H$. Indeed, consider two adjacent vertices $a \in V_1 \cap V(H)$ and $b \in V_2  \cap V(H)$. 
Suppose first that $a$ and $b$ belong to some $C_i$. Note that, with the exception of $y_i$ and maybe of the vertex $w_i$ (if it exists and $y_i \in Y$), every vertex of $C_i$ is $1$-monochromatic or $2$-monochromatic, the vertices of $V_1 \cap V(C_i)$ having even $2$-degree and the vertices of $V_2 \cap V(C_i)$ having odd $2$-degree.
Thus, no conflict involves two of these vertices. 
Suppose now that $b = y_i$. 
If $y_i$ is $2$-monochromatic with odd $2$-degree, then there is no conflict involving $y_i$ in $C_i$ since all of its neighbours in $C_i$ have even $2$-degree. If $y_i$ is special, then it is the only special vertex of $C_i$, so, here again, it cannot be involved in a conflict. 
If $y_i \notin Y$ and $y_i$ is $1$-monochromatic, then $y_i$ has no other $1$-monochromatic neighbour in $C_i$ by definition of $Y$. If $y_i \in Y$, then $y_i$ is $3$-monochromatic with odd $3$-degree, the only other possible $3$-monochromatic neighbour of $y_i$ in $C_i$ being $w_i$, but we showed previously that their $3$-degrees differ. Thus, in all cases, there cannot be conflicts between vertices of $C_i$. 

We are left with the case where $a$ and $b$ do not belong to the same $C_i$. In particular, this implies that $a \in X$ and that $a$ is $3$-monochromatic. The only possible $3$-monochromatic vertices in $V_2$ are the vertices of $Y$, which have odd $3$-degree, and the $3$-monochromatic vertices $u_1$ and $u_2$ with $3$-degree~$1$ and degree~$1$ in $H$ which might have been created at the very beginning of the proof. If $b \in Y$, then, due to the application of Lemma~\ref{lemma:bipartite} above, the only vertex of $X$ which can have odd $3$-degree is some $x_k$, but for this vertex we either ensured that it was involved in no conflict, or we tweaked the labelling so that it got even $3$-degree without modifying the labelling properties obtained through Lemma~\ref{lemma:bipartite}. If $b$ is $u_1$ or $u_2$, then $b$ has only one neighbour $v$. Note that the edges $vu_1$ and $vu_2$ are still labelled~$3$ as they are not part of the $Q_i$'s, and, thus, $d_3(b) = 1$ and $d_3(v) \geq 2$.
Hence, there is no conflict between vertices of $X$ and other vertices of $H$.
This implies that $H$ verifies Property~$(\mathcal P_3)$.
\end{proof}

\begin{claim}\label{claim:123:2neigh}
If $H$ contains a $1$-monochromatic vertex $u \in V_2$ with at least two neighbours in $H$, then we can relabel edges of $H$ so that $H$ verifies Property~$(\mathcal P_3)$.
\end{claim}

\begin{proof}
Let $v_1,\dots,v_p$ denote the neighbours of $u$ in $H$. Due to Lemma~\ref{lem:123} and because Claim~\ref{claim:123:noblue} does not apply on $H$,
for every vertex $v$ of $H$ we have $d_3(v) = 0$.
In particular, none of the $v_i$'s is $3$-monochromatic, implying that they are all $1$-monochromatic.  
Let $C_1, \dots, C_q$ be the $q \geq 1$ connected components of $H- u$. Every $C_i$ contains at least one of the $v_i$'s. Up to renaming the $v_i$'s, we can suppose w.l.o.g.~that $v_i \in V(C_i)$ if $i \leq q$. The vertices $v_i$ with $i > q$ (if any) can belong to any of the $C_i$'s.

Let us focus on one component $C_i$. Let $J^i_1, \dots, J^i_r$ denote the $r$ connected components of $C_i - v_i$. If $C_i$ has order~$1$, then by convention we set $r = 0$. In every $J^i_j$, choose a neighbour $x^i_j$ of $v_i$. 
By Lemma~\ref{lemma:bipartite}, we can relabel edges of $J^i_j$ with $1$ and~$2$ so that every vertex of $V_1 \cap V(J^i_j)$ has even $2$-degree, while every vertex of $V_2 \cap V(J^i_j)$, except possibly $x^i_j$, has odd $2$-degree.
Let $X_i$ be the set containing all $x^i_j$'s with even $2$-degree. Note that  $v_i$ has even $2$-degree, being precisely~$0$ since it is $1$-monochromatic; thus the only possible conflicts in $C_i$ involve vertices of $X_i$ as they are the only ones not following the parity rule on their $2$-degree (that is, they have even $2$-degree).

If $\abs{X_i} = 0$, $\abs{X_i} \geq 2$ or if $X_i = \set{w_i}$ and $d_2(w_i) \geq 1$ for some vertex $w_i$, then we say that $C_i$ is \emph{nice}. In this case, we can relabel edges of $C_i$ so that $C_i$ verifies Property~$(\mathcal P_3)$. If $\abs{X_i} = 0$, then $C_i$ already verifies Property~$(\mathcal P_3)$.
If $\abs{X_i} \geq 2$, then, for every $z \in X_i$, set $\ell(v_iz) = 3$. If $X_i = \set{w_i}$ and $d_2(w_i) \geq 1$, then set $\ell(v_iw_i) = 3$.
In the last two cases, all vertices of $X_i$ either become special while they have no special neighbours; or they become $3$-monochromatic with $3$-degree $1$ in which case $v_i$ is their only $3$-monochromatic neighbour and $d_3(v_i) \geq 2$. Moreover, in both cases, $d_3(v_i) \geq 1$ and all the neighbours of $v_i$ in $C_i$ which are not in $X_i$ have $3$-degree $0$. Thus, $v_i$ cannot be in conflict with its neighbours. Because the products of the other vertices of $C_i$ were not altered by these labelling modifications, $C_i$ verifies Property~$(\mathcal P_3)$.

If $X_i = \set{w_i}$ and $w_i$ is $1$-monochromatic with no such neighbours in $C_i - v_i$, then we say that $C_i$ is \emph{bad}. In such a bad component $C_i$, the only current conflict is between $v_i$ and $w_i$.
If $X_i = \set{w_i}$ and $w_i$ is $1$-monochromatic with at least one $1$-monochromatic neighbour $y_i$ in $C_i - v_i$, then we say that $C_i$ is \emph{tricky}. 
We denote by $N_n$ the number of nice components, by $N_b$  the number of bad components, and by $N_t$ the number of tricky components. 

In what follows, we consider several cases. In each case, we implicitly assume that none of the previous cases applies.

\begin{itemize}
    \item \textbf{Case 1.}  $N_t > 0$.
    
    Let $C_i$ be a tricky component. 
    For every bad or tricky component $C_j$ with $j \neq i$, set $\ell(v_jw_j) = 2$ and $\ell(uv_j) = 2$. In $C_j$, every vertex of $V_1$ now has even $2$-degree since $d_2(v_j) =2$ and every vertex of $V_2$ has odd $2$-degree since $d_2(w_j) =1$.
    
    Now, at this point:
    
    \begin{itemize}
        \item If $d_2(u)$ is even, then set $\ell(v_iw_i) = 2$ and $\ell(uv_i) = 2$. Here, $C_i$ behaves exactly like the other bad or tricky components and thus contains no conflicts.
        
        \item If $d_2(u)$ is odd, then set $\ell(v_iw_i) = \ell(w_iy_i) = 3$. Recall that all conflicts of $C_i$ involved $w_i$. Note that $w_i$ is now $3$-monochromatic with $3$-degree $2$ and no such neighbours. The vertices $y_i$ and $v_i$ are now $3$-monochromatic with $3$-degree $1$ and no such neighbours (in particular, they are not adjacent since they both belong to $V_1$). Hence $C_i$ does not contain any conflict.
    \end{itemize}
    
    In both cases, note that $u$ is $2$-monochromatic with odd $2$-degree. To summarise, we have reached the following situation.
    Special vertices (which were only created when dealing with nice components) only belong to $V_2$. $3$-monochromatic vertices are involved in no conflicts inside their component $C_j$ and have no $3$-monochromatic neighbours outside $C_j$ since $d_3(u) = 0$. All the other vertices of $H$ are either $1$-monochromatic or $2$-monochromatic: in particular, they have even $2$-degree if they belong to $V_1$, while they have odd $2$-degree if they belong to $V_2$. Hence, there is no conflict in $H$, and $H$ thus verifies Property~$(\mathcal P_3)$.
\end{itemize}

From now on, we can thus suppose that none of the $C_i$'s is tricky.

\begin{itemize}
    \item \textbf{Case 2.} $N_n = 0$.
    
    In this case, all $C_i$'s are bad. 
    We consider two cases:
    
    \begin{itemize}
        \item If $N_b=1$, \textit{i.e.}, $H$ contains only one (bad) component $C_1$, then set $\ell(v_1w_1) = 2$ and $\ell(uv_1) = 2$. Then every vertex of $H$ in $V_1$ is $1$-monochromatic or $2$-monochromatic with even $2$-degree, while every vertex in $V_2$ is $2$-monochromatic with odd $2$-degree. In particular, $d_2(w_1) = 1$, $d_2(v_1) = 2$ and $d_2(u) = 1$.
        
        \item If $N_b > 1$, then, for every (bad) component $C_i$, set $\ell(uv_i) = 3$. 
        Note that this makes all vertices of $H$ be monochromatic. 
        Every neighbour $z$ of $u$ verifies $d_3(z) \leq 1$ and, because $d_3(u) \geq 2$, the vertex $u$ cannot be in conflict with any of its neighbours in $H$. The vertices $v_i$ with $i \leq q$ are $3$-monochromatic with $3$-degree~$1$ and have no such neighbours. The $w_i$'s are $1$-monochromatic and have no $1$-monochromatic neighbours since the $C_i$'s were bad and their $v_i$'s (with $i \leq q$) are no longer $1$-monochromatic. 
        The other $1$-monochromatic vertices and $2$-monochromatic vertices raise no conflicts since, for every such vertex $z$ in $V_j \cap V(H)$ (where $j \in \set{1,2}$), we have  $d_2(z) \equiv j-1 \bmod 2$.  
    \end{itemize}
    
    Hence $H$ verifies Property~$(\mathcal P_3)$ in both cases. Thus, we can now assume $N_n>0$.

    \item \textbf{Case 3.} $N_{b} >0$. 
        
    Suppose now that at least one of the $C_i$'s is bad.
    Since  $N_n \geq 1$, not all $C_i$'s are bad. So, since $N_t=0$, we can thus suppose that $C_1$ is nice.
    For every bad component $C_j$, set $\ell(v_jw_j) = 2$ and $\ell(uv_j) = 2$. In $C_j$, every vertex of $V_1$ has even $2$-degree (since $d_2(v_j) =2$) while every vertex of $V_2$ has odd $2$-degree (since $d_2(w_j) =1$). 
    
    Let us now analyse the $2$-degree of $u$, which is at least~$1$ since $N_b>0$.
    
    \begin{itemize}
        \item If $d_2(u)$ is odd, then we claim that we have no conflicts in $H$. First, we saw earlier that any two vertices in a nice $C_i$ cannot be in conflict. Next, in every bad $C_j$, every vertex of $V_1$ has even $2$-degree, while every vertex of $V_2$ has odd $2$-degree; hence, any two vertices of $C_j$ cannot be in conflict. Thus, every possible conflict in $H$ must involve~$u$. Note that $u$ is $2$-monochromatic with odd $2$-degree while no vertex of $V_1 \cap V(H)$ is $2$-monochromatic with odd $2$-degree. Thus $u$ cannot be in conflict with a vertex of $H$.
    
        \item If $d_2(u)$ is even (and thus at least~$2$ since $N_b > 0$), then set $\ell(uv_1) = 3$. 
        Again, for the same reasons as earlier, any two  vertices in a $C_i$ with $i > 1$ cannot be in conflict. Since only $v_1$ had its product changed in $C_1$, then, if there is a conflict between two vertices of $C_1$, then it must involve $v_1$. Note that $v_1$ is $3$-monochromatic. If $d_3(v_1) \geq 2$, then it is the only vertex of $C_1$ with this property. If $d_3(v_1) = 1$, then $v_1$ was $1$-monochromatic before $uv_1$ was assigned labelled~$3$, in which case $v_1$, now, still has no $3$-monochromatic neighbours in $C_1$ by construction. Thus, in both cases, $v_1$ cannot be in conflict with any other vertex of $C_1$. Thus, any conflict in $H$ must involve $u$. Note that $u$ is special and that every other special vertex of $H$ must belong to some nice component 
        $C_i$, and must be a neighbour of $v_i$. In other words, all special vertices of $H$ must belong to $V_2$, and thus $u$ cannot be involved in a conflict.
    \end{itemize}

    Thus, in both cases, $H$ verifies Property~$(\mathcal P_3)$, and, from now on, we can assume $N_b=0$. That is, all $C_i$'s are nice, since also $N_t=0$.

    \item \textbf{Case 4.} $N_{n} = 1$.
    
    Since $N_b = N_t = 0$, we have that $H-u$ is connected, \textit{i.e.}, $q=1$ and $C_1$ is the only (nice) component. As we assumed that $d(u) \geq 2$, vertex $u$ has at least one other neighbour $v_2$ (in $V_1$) in $C_1$. Since $C_1$ is nice, recall that any two adjacent vertices of $C_1$ cannot be in conflict, due to how $\ell$ was modified so far.
    
    Let us analyse the possible situations, with respect to $v_1$.
    
    \begin{itemize}
        \item If $v_1$ is $1$-monochromatic, then set $\ell(uv_1) = \ell(uv_2) = 3$. In this case, $u$ has $3$-degree~$2$ while no other vertex of $H$ has $3$-degree at least~$2$. In $C_1$, the vertices of $V_1$ are either $2$-monochromatic with even $2$-degree, $1$-monochromatic, special (only $v_2$ can verify this, and this is only if $d_2(v_2) > 0$ since $d_2(v_2)$ is even), or $3$-monochromatic with $3$-degree~$1$ (only $v_1$ and $v_2$ can verify this, and, for the latter vertex, this is only if $d_2(v_2)=0$). Also, in $C_1$, the vertices of $V_2$ are $2$-monochromatic with odd $2$-degree. Hence, there are no conflicts.
        
        \item If $v_1$ is $3$-monochromatic, then set $\ell(uv_1) = 3$. In this case, in $H$, the vertices of $V_1$ are either $2$-monochromatic with even $2$-degree, $1$-monochromatic, or $3$-monochromatic with $3$-degree at least~$2$ (only $v_1$ verifies this). The vertices of $V_2$ are either $2$-monochromatic with odd $2$-degree, special, or $3$-monochromatic with $3$-degree~$1$ (in particular, $u$ verifies this). Hence, again there are no conflicts.
    \end{itemize}
        
    Thus, in both cases, $H$ eventually verifies Property~$(\mathcal P_3)$. From now on, in the next cases, we can thus assume that $N_n>1$.

    \item \textbf{Case 5.} $N_{n} \geq 2$ and there is some nice $C_i$ with $d_3(v_i) \geq 2$ that contains another neighbour $x$ of $u$ (\textit{i.e.}, $u$ has at least two neighbours in $C_i$).
    
    Assume $C_i$ does verify these properties.
    Let us start by modifying $\ell$, by changing to~$2$ the label assigned to every edge incident to $v_i$ assigned label~$3$. Note then that, in $C_i$, due to why we originally assigned label~$3$ to edges incident to $v_i$ in the first place, now every vertex of $V_2$ is $2$-monochromatic with odd $2$-degree while every vertex of $V_1 \setminus \set{v_i}$ is $2$-monochromatic with even $2$-degree. Also, due to our assumption on $v_i$, we have $d_2(v_i) \geq 2$.
    
    Let us now focus on $v_i$.
    
    \begin{itemize}
        \item If $d_2(v_i)$ is odd, then set $\ell(uv_i) = 2$. This makes $v_i$ become $2$-monochromatic with even $2$-degree with no such neighbours, while $u$ becomes $2$-monochromatic with odd $2$-degree with no such neighbours. 
    
        \item Assume now $d_2(v_i)$ is even.
        Let $C$ be a shortest cycle containing $u$, $v_i$ and $x$ (note that $C$ must exist since $C_i$ is connected). Now relabel every edge of $C$ so that $1$'s becomes $2$'s and \textit{vice versa}.  Note that, as a result, we get $d_2(u) = 2$, and, in $C_i$, every vertex of $V_2$ is $2$-monochromatic with odd $2$-degree while every vertex of $V_1$ is $1$-monochromatic or $2$-monochromatic with even $2$-degree. Hence, there is no conflict in $C_i$. Also, since every $C_j$ with $j \neq i$ is nice and we did not modify labels incident to vertices of $C_j$, there are still no conflicts in $C_j$.
        
        If no conflicts remain, then $H$ now verifies Property~$(\mathcal P_3)$. So assume some conflicts remain. All these conflicts must involve $u$, but, now, we have that $d_2(u)=2$.
        Since $N_n \geq 2$, there exists $v_j \notin C_i$ and $j \leq q$. Set $\ell(uv_j) = 3$, so that $u$ becomes special. Note that this increases $d_3(v_j)$. If $v_j$ had $3$-degree~$0$, then $v_j$ was $1$-monochromatic and $C_j$ had no $3$-monochromatic vertices, and, hence, now, there is no conflict in $C_j$.  If $v_j$ had non-zero $3$-degree, then every neighbour of $v_j$ in $C_j$ still has $3$-degree at most~$1$ while $v_j$ has $3$-degree at least $2$. Hence, there is no conflict in $C_j$. 
        
        From here, it can be checked that no conflicts remain at all in $H$. In particular, all special vertices, including $u$, lie in $V_2$, and they are thus not in conflict. Thus, $u$ is not in conflict. Also, there is still no conflict in a $C_k$ with $k \notin{i,j}$ since $C_k$ is nice and the products of their vertices did not change. Also, there is no conflict in $C_i$ and $C_j$ by our previous remarks. 
    \end{itemize}

    Thus, in both cases, $H$ verifies Property~$(\mathcal P_3)$. We now deal with a final case.
    
    \item \textbf{Case 6.} $N_{n} \geq 2$.
    
    Let $A = \set{a_1,\dots,a_r}$ be the subset of neighbours of $u$ having $2$-degree~$0$. Note that $r \geq N_n \geq 2$. Note also that some of these $a_i$'s are $v_i$'s with $i\leq q$ (all of which are in nice components, since $N_t=N_b=0$), in which case, by how the nice components were treated earlier, they can be $3$-monochromatic. Furthermore, $A$ may contain more than $N_{n}$ vertices since it may also contain $1$-monochromatic $v_i$'s with $2$-degree~$0$ and $i > q$. However, since the previous case does not apply, if some $v_i$ verifies $d_3(v_i) \geq 2$ (thus $i \leq q$), then $u$ cannot neighbour any other vertex of $C_i$.
    
    For every $a_i \in A$, we define $n_i$ as the current value of $d_3(a_i)$, at the beginning of this case (\textit{i.e.}, before modifying labels below). Recall that we can have $d_3(a_i)>0$, in which case $a_i$ is a $v_j$ in a (nice) $C_j$ for which we had to remove some conflicts. Also, by the choice of $A$, at this point, $\ell(ua_i) = 1$. The goal now, is to relabel some $ua_i$'s with $3$ in such a way that $u$ is not in conflict with the vertices of $A$. To show this can be achieved, we use the \CombNull{}~\cite{Alo99}.
    
    \begin{theorem}[\CombNull{}~\cite{Alo99}]
	\label{thm:null} \label{th:combi-null}
	Let $\mathbb{F}$ be an arbitrary field, and $P=P(Z_1,\dots,Z_p)$ be a polynomial in $\mathbb{F}[Z_1,\dots,Z_p]$.
	Suppose that the coefficient of a monomial $Z_1^{k_1}\dots Z_p^{k_p}$, where every $k_i$ is a non-negative integer, is non-zero in $P$ and the degree of $P$ equals $\sum_{i=1}^p k_i$.
	If $S_1,\dots,S_p$ are subsets of $\mathbb{F}$ with $|S_i|>k_i$ for every $i \in \{1,\dots,p\}$,
	then there are $z_1\in S_1,\dots,z_p\in S_p$ so that $P(z_1,\dots,z_p) \neq 0$.
    \end{theorem}
    
    For every $i \in \{1,\dots,r\}$, let $Z_i$ be a variable belonging to $S_i = \set{0,1}$ and representing whether $ua_i$ is assigned label~$3$ ($Z_i = 1$) or not ($Z_i = 0$).
    Let $P$ be the following polynomial:
    $$ P(Z_1,\dots,Z_r) = \prod_{i=1}^r \paren{\sum_{\substack{j=1\\j\neq i}}^r Z_j - n_i}.$$
    Since $r \geq N_{n} \geq 2$, note that $P$ has degree $r$ at least~$2$. Furthermore, the monomial $\prod_{i=1}^r Z_i$ has non-zero coefficient (since every $Z_i$ has positive coefficient in the description of~$P$). Hence the \CombNull{} applies and there is a way to choose values $z_1,\dots,z_r$ in $\{0,1\}$ for $Z_1,\dots,Z_r$ so that $P(z_1,\dots,z_r) \neq 0$.
    
    Now, for every $i \in \{1,\dots,r\}$ for which $z_i = 1$, set $\ell(ua_i) = 3$. Note that $d_3(u) =  \sum_{j=1}^r z_j$ and $d_2(u) = 0$.
    We claim that $H$ now verifies Property~$(\mathcal P_3)$. Assume this is wrong, and suppose that there is a conflict in $H$ between two vertices $x \in V_1$ and $y \in V_2$. For now, suppose that $u$ is not one of these two vertices.
    
    \begin{itemize}
        \item If $x$ and $y$ are $2$-monochromatic, then, because we did not modify $2$-degrees when we modified $\ell$ above, and all $C_i$'s are nice, then $d_2(x)$ is even while $d_2(y)$ is odd, a contradiction to the fact that $x$ and $y$ are in conflict. 
        
        \item Similarly, the modifications above did not introduce new $1$-monochromatic vertices. Thus, $x$ and $y$ cannot be both $1$-monochromatic, since all $1$-monochromatic vertices of $H$ (different from $u$) belong to $V_1$.
        
        \item Similarly, $x$ and $y$ cannot be special. This is because, since the $a_i$'s have $2$-degree $0$, the modifications did not introduce new special vertices. So, all special vertices are adjacent to $v_i$'s (with $i \leq q$), and thus lie in $V_2$.
        
        \item If $x$ and $y$ are $3$-monochromatic, then $y$ must be a neighbour of some $v_i$ (with $i \leq q$) and $y$ thus verifies $d_3(y) = 1$. In this case, $v_i$ verified $d_3(v_i) \geq 2$ at the beginning of this case (by how $\ell$ was constructed in $C_i$, and, in particular, because $y$ is not special), and thus $x \neq v_i$. Thus, $x$ became $3$-monochromatic because $ux$ was relabelled with label~$3$ through the \CombNull{}. So we deduce that $u$ has two neighbours in $C_i$, where we had $d_3(v_i) \geq 2$ at the beginning of this case. This is not possible, as this configuration is forbidden due to previous \textbf{Case~5} not applying.
    \end{itemize}

    Hence, every possible conflict must involve $u$.
    Vertex $u$ has two types of neighbours: those with non-zero $2$-degree, and the vertices of $A$. Since $d_2(u) = 0$, the first group of neighbours cannot be in conflict with $u$.  Suppose now that $a_i \in A$ is in conflict with $u$. Note that $d_3(a_i) = n_i + z_i$ and $d_3(u) =  \sum_{j=1}^r z_j$. Since $d_3(a_i) = d_3(u)$, we have $\sum_{\substack{j=1\\j\neq i}}^r z_j - n_i =0$ and thus $P(z_1,\dots,z_r) = 0$, a contradiction.
    
    Hence there is no conflict in $H$, and $H$ verifies Property~$(\mathcal P_3)$.\qedhere
\end{itemize}
\end{proof}

We are now ready to get rid of the last possible conflicts in $\mathcal H$.

\begin{claim}\label{claim:123:lastcase}
For every remaining $H$,
we can relabel edges so that $H$ verifies Property~$(\mathcal P_3)$.
\end{claim}

\begin{proof}
Let $v \in V_1$ and $u \in V_2$ be two adjacent $1$-monochromatic vertices of $H$ (which must exist as otherwise $H$ would verify Property~$(\mathcal P_3)$). Because $H$ has at least two edges (as otherwise it would belong to $M$, not to $\mathcal{H}$), at least one of $v$ and $u$ must have another neighbour in $H$. 
Since Claim~\ref{claim:123:noblue} does not apply, the neighbours of $u$ are $1$-monochromatic and since Claim~\ref{claim:123:2neigh} does not apply, $u$ must have degree~$1$ in $H$. So $v$ is also adjacent to $k \geq 1$ vertices $x_1,\dots,x_k \in V_2$ different from $u$. Still by Claim~\ref{claim:123:2neigh}, note that if an $x_i$ is $1$-monochromatic, then it must be of degree~$1$ in $H$, since $v$ is a neighbour of $x_i$; but then Claim~\ref{claim:123:noblue} would apply, as $v$ is $1$-monochromatic and neighbours $u$ and $x_1$, which are $1$-monochromatic and of degree~$1$ in $H$. Thus, we can assume all $x_i$'s are $2$-monochromatic (because of incident downward edges to $V_3, \dots, V_t$; recall that all edges of $H$ are labelled~$1$).


Set $H'=H-u$.
According to Lemma~\ref{lemma:bipartite}, we can relabel edges in $H'$ with~$1$ and~$2$ so that all vertices in $(V_1 \cap V(H')) \setminus \{v\}$ have odd $2$-degree, while all vertices in $V_2 \cap V(H')$ have even $2$-degree. Recall that $u$ is $1$-monochromatic. Thus, if also $v$ is $2$-monochromatic with odd $2$-degree, then we are done. Assume thus that $v$ is $2$-monochromatic with even $2$-degree.

\begin{itemize}
    \item Assume first that the $2$-degree of $v$ is even at least~$2$. In that case, set $\ell(vu)=3$. This way, $u$ becomes $3$-monochromatic, while $v$ becomes special.
    
    \item Assume now $v$ is $1$-monochromatic. This implies that $\ell(vx_1)=1$. Change $\ell(vx_1)$ to $3$. This way, $x_1$ becomes special (recall its $2$-degree is even and at least~$1$, due to incident downward edges), while $v$ becomes $3$-monochromatic. Note that $u$ remains $1$-monochromatic.
\end{itemize}

In both cases, it can be checked that $H$ now fulfils Property~$(\mathcal P_3)$.
\end{proof}

At this point, we dealt with all connected components of $\mathcal{H}$, and the resulting labelling $\ell$ of $G$ is p-proper by Lemma~\ref{lem:123:P3}.
The whole proof is thus complete.

\section{Conclusion}\label{section:ccl}

Although we provide a solution to the product version of the 1-2-3 Conjecture, our investigations and our proof methodology actually open the way to several appealing directions for further research on the topic. In particular:

\begin{itemize}
    \item Distinguishing labellings, generally speaking, is a field with many interconnections between more or less distant problems, and, as a result, any major breakthrough on one particular distinguishing labelling problem can have drastic consequences on related others. A perfect illustration for that claim, is through the example of a brilliant algorithm designed by Kalkowski in~\cite{Kal09} to get very close to a full verification of the total version of the 1-2-3 Conjecture (where vertices are also labelled, the label assigned to every vertex taking part to its sum) from~\cite{PW10}. Since its introduction, Kalkowski's Algorithm has indeed been revisited in numerous works, which, sometimes, allowed to improve significantly the best results that were known for long. In particular, the upper bound, from~\cite{KKP10}, of $5$ on $\chis(G)$ for every nice graph $G$ results from straight modifications of Kalkowski's Algorithm. Very interesting results for generalisations of the 1-2-3 Conjecture to hypergraphs were also established through modifications of Kalkowski's approach~\cite{KKP17}. In~\cite{KKP11}, new bounds on the irregularity strength of graphs (which is, roughly put, a generalisation of the 1-2-3 Conjecture where all vertices, not only the adjacent ones, are required to be distinguished through their sums by a labelling) were established, and the proof arguments were, again, strongly influenced by Kalkowski's Algorithm. Distinguishing labellings really form a field where making significant progress relies on the introduction of novel ideas, which might lead to many appealing perspectives for the whole field.
    
    As seen through this work, the product version of the 1-2-3 Conjecture, and in particular p-proper labellings, rely on very peculiar properties. Yet, proving it required quite some efforts, the resulting proof being rather technical at times. As mentioned in the introductory section, we were highly influenced by Vu\v{c}kovi\'c result from~\cite{Vuc18}, which, we believe, is another one of these major results that can lead to many interesting accomplishments, as our main result in this work just showcases.
    
    According to these thoughts, one can naturally wonder whether our proof scheme could in turn be modified to deal with problems that are close to the product version of the 1-2-3 Conjecture. A few candidates come immediately to mind. In particular, one could wonder whether we can get new results on the product irregularity strength of graphs~\cite{Anh09} (in which all vertices must be distinguished through their products by a labelling). One could also wonder about consequences for the list version of the product version of the 1-2-3 Conjecture (introduced in~\cite{BHLS22}, in which labellings must be constructed by assigning labels from dedicated lists of three labels). We are not sure exactly what one could expect, but these questions would definitely be worth considering.
    
    \item Note that an m-proper $3$-labelling is similar, when $a,b,c$ are pairwise coprime labels, to a p-proper $\{a,b,c\}$-labelling. Thus, for any three pairwise coprime labels $a,b,c$, the result of Vu\v{c}kovi\'c from~\cite{Vuc18} implies that every nice graph admits a p-proper $\{a,b,c\}$-labelling. An intermediate question lying in between the product version of the 1-2-3 Conjecture and its list variant would thus be about the existence of p-proper $\{a,b,c\}$-labellings for any nice graph and any three fixed labels $a,b,c$.
    
    \item Other directions of interest would deal with the connections between the sum version and the product version of the 1-2-3 Conjecture. Note indeed that there are definitely connections, as, by labellings, label~$0$ in the sum version plays the same role as label~$1$ in the product version. For this reason, s-proper $\{0,a\}$-labellings and p-proper $\{1,b\}$-labellings are similar objects for any $a,b \neq 0$. When considering three labels, note that the situation is not as obvious, as the equivalence between an s-proper $\{0,a_1,a_2\}$-labelling and a p-proper $\{1,b_1,b_2\}$-labelling is not guaranteed (as two sums of $a_1$'s and $a_2$'s might be different while the corresponding two products of $b_1$'s and $b_2$'s might not be, and \textit{vice versa}). However, there are situations where this is guaranteed, for instance when $a_i = \log(b_i)$ for every $i \in \{1,2\}$, or when the $a_i$'s (and $b_i$'s) are such that we can infer the coefficients of a sum (and product) from said sum (and product). 
    
    These thoughts relate to an interesting question related to the 1-2-3 Conjecture. By the arguments above, it can be checked that from an m-proper $3$-labelling of some graph $G$, we can obtain an s-proper $\{1,\Delta(G),\Delta(G)^2\}$-labelling of $G$. So, for $G$, there indeed exist three labels $a_G,b_G,c_G$ for which we know s-proper $\{a_G,b_G,c_G\}$-labellings exist. Note however that these $a_G,b_G,c_G$ are functions of $G$, and thus, for a graph $H$ different from $G$, we might have $\{a_G,b_G,c_G\} \neq \{a_H,b_H,c_H\}$. The question is whether we can provide three labels $a^*,b^*,c^*$ that would work for all nice graphs. The 1-2-3 Conjecture asserts that $1,2,3$ would be an example of three such labels $a^*,b^*,c^*$.
    As mentioned earlier, the result of Vu\v{c}kovi\'c implies that there is an s-proper $\{1,b_G,c_G\}$-labelling of any nice graph $G$ where $b_G$ and $c_G$ are functions of $G$. By earlier arguments, our proof of the product version of the 1-2-3 Conjecture implies that that there is an s-proper $\{0,1,c_G\}$-labelling of every nice graph $G$, where $c_G$ is a function of $G$. Thus, in some sense, we are now just one step away from providing three labels $a^*,b^*,c^*$ as described above.
\end{itemize}

\section*{Acknowledgement}

The authors are grateful to the three anonymous referees for their careful reading of a previous version of the current work, which 
allowed to improve the general quality and correctness not only of  the main proof, but also of the whole paper.


\begin{thebibliography}{99}

\bibitem{AADR05} L. Addario-Berry, R.E.L. Aldred, K. Dalal, B.A. Reed. Vertex colouring edge partitions. \emph{Journal of Combinatorial Theory, Series B}, 94(2):237-244, 2005.

\bibitem{Alo99} N. Alon. Combinatorial Nullstellensatz. \textit{Combinatorics, Probability and Computing}, 8:7-29, 1999.

\bibitem{Anh09} M. Anholcer. Product irregularity strength of graphs. \textit{Discrete Mathematics}, 309(22):6434-6439, 2009.


\bibitem{BHLS20} J. Bensmail, H. Hocquard, D. Lajou, \'E. Sopena. Further Evidence Towards the Multiplicative 1-2-3 Conjecture. \emph{Discrete Applied Mathematics}, 307:135-144, 2022.

\bibitem{BHLS22} J. Bensmail, H. Hocquard, D. Lajou, É. Sopena. On a List Variant of the Multiplicative 1-2-3 Conjecture. \textit{Graphs and Combinatorics}, 38(3):88, 2022.

\bibitem{Gal98} J.A. Gallian. A dynamic survey of graph labeling. \emph{Electronic Journal of Combinatorics}, 6, 1998.


\bibitem{Kal09} M. Kalkowski. A note on the 1,2-Conjecture. Ph.D. thesis, Adam Mickiewicz University, Poland, 2009.

\bibitem{KKP10} M. Kalkowski, M. Karo\'nski, F. Pfender. Vertex-coloring edge-weightings: towards the 1-2-3 Conjecture. \textit{Journal of Combinatorial Theory, Series B}, 100:347-349, 2010.

\bibitem{KKP11} M. Kalkowski, M. Karo\'nski, F. Pfender. A new upper bound for the irregularity strength of graphs. \textit{SIAM Journal of Discrete Mathematics}, 25(3):1319-1321, 2011.

\bibitem{KKP17} M. Kalkowski, M. Karo\'nski, F. Pfender. The 1‐2‐3‐Conjecture for Hypergraphs. \textit{Journal of Graph Theory}, 85(3):706-715, 2017.

\bibitem{KLT04} M. Karo\'nski, T. {\L}uczak, A. Thomason. Edge weights and vertex colours. \textit{Journal of Combinatorial Theory, Series B}, 91:151–157, 2004.

\bibitem{Prz19c} J. Przyby{\l}o. The 1-2-3 Conjecture almost holds for regular graphs. \textit{Journal of Combinatorial Theory, Series B}, 147:183-200, 2021.

\bibitem{PW10} J. Przyby{\l}o, M. Wo\'zniak. On a 1,2 Conjecture. \textit{Discrete Mathematics and Theoretical Computer Science}, 12(1):101-108, 2010.


\bibitem{SK12} J. Skowronek-Kazi\'ow. Multiplicative vertex-colouring weightings of graphs. \textit{Information Processing Letters}, 112(5):191-194, 2012.


\bibitem{Vuc18} B. Vu\v{c}kovi\'c. Multi-set neighbor distinguishing $3$-edge coloring. \textit{Discrete Mathematics}, 341:820-824, 2018.



\end{thebibliography}
\end{document}